\documentclass[11pt]{article}

\usepackage[utf8]{inputenc}
\usepackage[T1]{fontenc}
\usepackage{lmodern}
\usepackage[protrusion=true,expansion=false]{microtype}
\usepackage[margin=1in, top=1.1in, bottom=1.1in]{geometry}

\usepackage{amsmath,amssymb,amsthm}

\usepackage{booktabs}
\usepackage{graphicx}

\usepackage[colorlinks=true,
            linkcolor=blue!70!black,
            urlcolor=blue!70!black,
            citecolor=blue!70!black]{hyperref}

\usepackage{xcolor}
\usepackage{enumitem}
\usepackage[numbers,sort&compress]{natbib}
\usepackage{tikz}
\usetikzlibrary{arrows.meta, positioning, fit, backgrounds, patterns, decorations.pathreplacing}

\usepackage{titlesec}
\titleformat{\section}{\large\bfseries}{\thesection}{1em}{}
\titleformat{\subsection}{\normalsize\bfseries}{\thesubsection}{1em}{}

\theoremstyle{plain}
\newtheorem{theorem}{Theorem}[section]

\newtheorem{corollary}[theorem]{Corollary}

\theoremstyle{definition}
\newtheorem{definition}[theorem]{Definition}
\newtheorem{example}[theorem]{Example}
\newtheorem{assumption}[theorem]{Assumption}

\theoremstyle{remark}
\newtheorem{remark}[theorem]{Remark}

\newcommand{\Allow}{\textsc{Allow}}
\newcommand{\Refuse}{\textsc{Refuse}}
\newcommand{\Escalate}{\textsc{Escalate}}
\newcommand{\Dset}{\mathcal{D}}
\newcommand{\Adm}{\mathrm{Adm}}
\newcommand{\Req}{\mathrm{Req}}
\newcommand{\Traces}{\mathrm{Traces}}
\newcommand{\eval}{\mathrm{eval}}
\newcommand{\exec}{\mathrm{exec}}

\title{Atomic Decision Boundaries: A Structural Requirement\\
for Guaranteeing Execution-Time Admissibility in Autonomous Systems}

\author{Marcelo Fernandez\\
TraslaIA\\
\texttt{info@traslaia.com}}

\date{April 2026 \\ \small \href{https://arxiv.org/abs/2604.17511}{arXiv:2604.17511 [cs.AI]} \quad \href{https://doi.org/10.5281/zenodo.19670649}{DOI: 10.5281/zenodo.19670649}}

\begin{document}
\maketitle

\begin{abstract}
Autonomous systems increasingly execute actions that directly modify shared state,
creating an urgent need for precise control over which transitions are permitted
to occur.  Existing governance mechanisms evaluate policies prior to execution or
reconstruct behavior post hoc, but do not enforce admissibility at the exact
moment a state transition is committed.

We introduce the \emph{atomic decision boundary}, a structural property of
admission control systems in which the decision and the resulting state
transition are jointly determined as a single indivisible step.  Formalizing
execution as a labeled transition system (LTS), we distinguish two classes:
\emph{atomic systems}, where evaluation and transition are coupled within a
single LTS step, and \emph{split evaluation systems}, where they are separate
transitions that may be interleaved by environmental actions.

Under realistic concurrent environments (Assumptions~2.1--2.2), we prove that
no construction can make a split system equivalent to an atomic system with
respect to admissibility under all execution traces.  This limitation is a
structural property of the system architecture, not of policy expressiveness
or state availability.  We further formalize the \Escalate{}
outcome---absent from classical time-of-check/time-of-use (TOCTOU) analyses---
and show that its resolution is itself subject to the atomic boundary
requirement, extending the guarantees of the main theorem.

We map widely used policy enforcement mechanisms, including RBAC and Open
Policy Agent~(OPA), to the split model, and contrast them with systems that
enforce atomic decision boundaries.  Our results establish that admissibility
is a property of execution, not of prior evaluation, and that guaranteeing it
requires structural guarantees at the decision boundary.
This paper is the formal foundation of a 6-paper Agent Governance Series.
Companion papers address stateful enforcement (ACP, Paper~1,
arXiv:2603.18829~\cite{acp26}), behavioral drift detection above the
enforcement boundary (IML, Paper~2,~\cite{iml26}), fair multi-agent allocation
(Paper~3,~\cite{fairgov26}), composition irreducibility
(Paper~4,~\cite{fernandez2026comp}), and runtime execution validity under
partial observability---RAM (Paper~5,~\cite{fernandez2026ram}).
\end{abstract}

\tableofcontents
\newpage

\section{Introduction}
\label{sec:intro}

Autonomous systems are increasingly deployed in settings where they plan and
execute multi-step action sequences that directly and irreversibly modify system
state.  Financial transfers, document operations, API calls, and inter-agent
delegations all commit changes that cannot simply be rolled back.  As a result,
the question of whether a given action should be allowed to execute must be
resolved precisely at the moment the transition is committed---not before, and
not after.

Consider a concrete scenario.  An agent is authorized to transfer funds from
a shared ledger.  A policy engine evaluates the request against the current
balance and account status, returns $\Allow$, and triggers the transfer.  In
the interval between that evaluation and the execution of the debit, a
concurrent transaction---another agent repaying a loan---reduces the balance
below the required amount.  The transfer executes against an inadmissible
state.  The guarantee the policy was meant to enforce has evaporated.  No
refinement of the policy expression could have prevented this: the failure
is not in what the policy \emph{knows}, but in \emph{when} the transition
fires relative to the state it governs.  We call the gap between these two
moments the \emph{decision boundary problem}, and formalize it precisely in
what follows.

Current governance approaches address this question indirectly.  \emph{Policy
engines} evaluate rules against input data prior to execution.  \emph{Access
control models} determine permissions based on static role or attribute
assignments.  \emph{Audit systems} reconstruct behavior after the fact.
These mechanisms are useful and widely deployed, but they share a structural
limitation: the decision to allow an action and the transition that applies
that action are separate operations.

This separation introduces a gap.  The decision is evaluated in one system
state; the transition fires in a---potentially different---state.  An action
admissible when evaluated may become inadmissible by execution time.  Under
the environment model we formalize in \S\ref{sec:model}, no policy, regardless
of its sophistication, can close this gap from within a split architecture:
the gap is not an information problem---it is an architectural one.

This paper formalizes that observation precisely.  We define the
\emph{decision boundary} of a governance system as the point at which the
system resolves whether a state transition $(s \to s')$ is permitted to exist.
We then define an \emph{atomic decision boundary} as one in which the
resolution and the resulting transition are jointly computed as a single
indivisible step in the labeled transition system model of execution.

We prove that systems lacking atomicity at the decision boundary---henceforth
\emph{split evaluation systems}---cannot guarantee admissibility at execution
time under all execution traces.  The impossibility is constructive: we exhibit
a concrete execution trace that witnesses the gap.

A distinguishing feature of our model is the three-valued decision domain
$\Dset = \{\Allow, \Refuse, \Escalate\}$.
The $\Escalate$ outcome---absent from classical TOCTOU analyses
\cite{bishop96, tsafrir08}---suspends the transition pending supervisor review,
a first-class governance action in modern agent systems \cite{acp26}.
We show that $\Escalate$ does not eliminate the atomicity requirement; it
\emph{transfers} it to the resolution step, and prove that the same
structural gap arises there unless resolution is itself atomic.

\paragraph{Contributions.}
\begin{enumerate}[label=\normalfont(\roman*)]
  \item We formalize the execution model for admission control as an LTS and
    define admissibility as a property of state transitions, not of actions
    in isolation (\S\ref{sec:model}).

  \item We define the \emph{atomic decision boundary} in LTS terms and
    contrast it with \emph{split evaluation systems} via a precise structural
    distinction (\S\ref{sec:dbm}).

  \item We give a complete formal semantics for the $\Escalate$ outcome,
    including the supervisor resolution function and the state it operates
    over (\S\ref{sec:dbm}).

  \item We prove, via a constructive counterexample trace, that no general
    construction can make a split system equivalent to an atomic system with
    respect to admissibility (\S\ref{sec:results}).  Three corollaries
    follow: impossibility of execution-time guarantees in split systems,
    insufficiency of external state, and admissibility as an execution-time
    property.

  \item We prove the \emph{escalation closure requirement}: the $\Escalate$
    outcome preserves admissibility guarantees if and only if the supervisor
    resolution is itself atomic (\S\ref{sec:results}).

  \item We map RBAC, OPA, and ACP to the two classes, providing a structural
    taxonomy of existing governance mechanisms (\S\ref{sec:mapping}).
\end{enumerate}

This paper is Paper~0 in a series of six.  The Agent Control Protocol~(ACP)~\cite{acp26}
(Paper~1) instantiates the atomic decision boundary for multi-agent governance.
The Invariant Measurement Layer~(IML)~\cite{iml26} (Paper~2) operates above the
atomic enforcement layer and addresses behavioral drift that remains invisible
to enforcement-based monitoring.  Paper~3~\cite{fairgov26} addresses the fair
allocation of atomic governance decisions across agents operating under shared
resource constraints.  Paper~4~\cite{fernandez2026comp} proves the irreducibility
of the four-layer governance architecture.  Paper~5~\cite{fernandez2026ram}
closes the series by providing the runtime operational mechanism: the
Reconstructive Authority Model (RAM), which determines whether execution is
valid at every step under partial observability.

\section{System Model}
\label{sec:model}

We model the execution environment as a labeled transition system and derive
the formal definition of admissibility from it.  The section is organized into
two parts: an execution model (\S\ref{sec:execmodel}) that establishes the LTS
and the admissibility predicate, and an environment and threat model
(\S\ref{sec:envmodel}) that specifies the assumptions under which the main
theorem is stated.

\subsection{Execution Model}
\label{sec:execmodel}

\begin{definition}[Labeled Transition System]
\label{def:lts}
A \emph{labeled transition system} (LTS) is a triple $M = (S, \mathrm{Act}, {\to})$ where:
\begin{itemize}
  \item $S$ is a set of \emph{states};
  \item $\mathrm{Act} = A \cup A_{\mathrm{env}}$ is a set of \emph{action labels},
    partitioned into \emph{agent actions} $A$ (subject to governance) and
    \emph{environment actions} $A_{\mathrm{env}}$ (external state modifications),
    with $A \cap A_{\mathrm{env}} = \emptyset$;
  \item ${\to} \subseteq S \times \mathrm{Act} \times S$ is the
    \emph{transition relation}.
\end{itemize}
We write $s \xrightarrow{a} s'$ for $(s, a, s') \in {\to}$.
\end{definition}

\begin{definition}[Execution Trace]
\label{def:trace}
An \emph{execution trace} $\sigma$ is a finite alternating sequence
$s_0\, a_1\, s_1\, a_2\, s_2 \cdots a_n\, s_n$
such that $s_{i-1} \xrightarrow{a_i} s_i$ for all $i \in \{1, \ldots, n\}$.
We denote the set of all execution traces of $M$ by $\Traces(M)$.
\end{definition}

\begin{definition}[Admissibility]
\label{def:adm}
An \emph{admissibility predicate} is a function
$\Adm : S \times A \to \{\mathsf{true}, \mathsf{false}\}$.
Action $a \in A$ is \emph{admissible} in state $s \in S$ if $\Adm(s, a) = \mathsf{true}$.
Admissibility is a property of the state at the moment of transition, not of
the action in isolation.
\end{definition}

\begin{definition}[Admissibility-Preserving System]
\label{def:admpres}
A system $M$ \emph{preserves admissibility} if for every trace
$\sigma \in \Traces(M)$ and every agent action $a \in A$ appearing in $\sigma$,
whenever $s \xrightarrow{a} s'$ occurs in $\sigma$, we have $\Adm(s, a) = \mathsf{true}$.
\end{definition}

\paragraph{The evaluation--execution gap.}
In many practical systems, deciding whether to allow an action and applying
the resulting transition are implemented as distinct operations.  Let $D : S
\times A \to \Dset$ denote the \emph{decision function} that evaluates
admissibility, and let $T : S \times A \to S$ denote the \emph{transition
function} that applies it.  In split systems, $D$ and $T$ are invoked
separately.  The state in which $D$ is invoked, call it $s_{\eval}$, need not
equal the state in which $T$ fires, call it $s_{\exec}$.  This is the
structural source of the problem we formalize.

\subsection{Environment and Threat Model}
\label{sec:envmodel}

The main theorem (Theorem~\ref{thm:main}) is a conditional result: it holds
under a specific model of the execution environment.  We state this model
explicitly here, as two assumptions, so that the scope of the result is
unambiguous.  The environment model is \emph{adversarial}: the environment is
assumed capable of interleaving state-modifying actions at the worst possible
moment, with no synchronization guarantee for the governance system.  This is
the standard assumption in concurrent systems analysis and corresponds to the
worst-case deployment context for multi-agent systems operating over shared
mutable state.  Readers who restrict to sequential or single-threaded settings
should consult Remark~\ref{rem:scope} for the precise boundary conditions under
which the theorem applies.

\begin{assumption}[Non-Triviality]
\label{asm:nontrivial}
The system satisfies the following:
\begin{enumerate}[label=(\roman*)]
  \item There exist $s \in S$ and $a \in A$ with $\Adm(s, a) = \mathsf{true}$
    and $D(s, a) = \Allow$ (the decision function grants at least one
    admissible action without escalation).
  \item There exist $s \in S$, $a \in A$, and $e \in A_{\mathrm{env}}$ such that
    $s \xrightarrow{e} s^*$ and $\Adm(s^*, a) = \mathsf{false}$.
\end{enumerate}
\end{assumption}

Assumption~\ref{asm:nontrivial} is satisfied by any system where admissibility
genuinely depends on state.  A governance mechanism over a state-independent
system provides no meaningful constraint; a system whose state cannot change
between decision and execution is already effectively atomic.  In practice,
every multi-agent system operating over shared mutable state---including file
systems, API quota registries, agent delegation ledgers, and shared databases---
satisfies both conditions: actions can be admitted and state can be modified
by concurrent environment processes between evaluation and execution.

\begin{assumption}[Environment Model]
\label{asm:env}
The environment satisfies the following:
\begin{enumerate}[label=(\roman*)]
  \item \emph{Uncontrollability.}  Environment actions $e \in A_{\mathrm{env}}$
    are not subject to governance control: the admission system cannot prevent,
    delay, or reorder them.

  \item \emph{Arbitrary interleaving.}  In a split evaluation system, environment
    actions may fire at any point in the LTS between the decision transition
    $\mathit{dec}(a)$ and the execution transition $\mathit{exec}(a)$.  The
    system provides no synchronization barrier that would prevent such
    interleaving.

  \item \emph{State visibility.}  The decision function $D$ observes the system
    state at the time it is invoked.  It does not have predictive access to
    future states produced by environment actions that have not yet fired.
\end{enumerate}
\end{assumption}

Assumption~\ref{asm:env}(i) reflects the standard distinction between
controllable and uncontrollable events in supervisory control
theory~\cite{ramadge87}.  Assumption~\ref{asm:env}(ii) is the adversarial model
under which the theorem is stated: the environment is assumed capable of
interleaving at the worst moment.  Assumption~\ref{asm:env}(iii) is satisfied
by every decision function that operates on observable state; a function
with oracular access to future states would not correspond to any implementable
system.  The three assumptions hold collectively in all multi-agent systems
where agents share mutable state with concurrent external processes.

\begin{example}[Running Example]
\label{ex:running}
An agent manages file operations on a shared resource.  The state $s$ encodes
the set of locked files $\mathit{locked}(s) \subseteq \mathcal{F}$ and the
current quota usage $\mathit{quota}(s) \in \mathbb{N}$.  Action $\mathit{write}(f)$
is admissible in $s$ iff $f \notin \mathit{locked}(s)$ and
$\mathit{quota}(s) < \mathit{quota}_{\max}$.  Environment action
$\mathit{lock}(f)$ transitions any state $s$ to state $s^* = s[\mathit{locked}
\mapsto \mathit{locked}(s) \cup \{f\}]$, making $\mathit{write}(f)$ inadmissible.
Assumption~\ref{asm:nontrivial} is satisfied.
\end{example}

\section{The Decision Boundary Model}
\label{sec:dbm}

\subsection{Decision and Transition Functions}

\begin{definition}[Decision Domain]
Let $\Dset = \{\Allow, \Refuse, \Escalate\}$.  We call elements of $\Dset$
\emph{dispositions}.
\end{definition}

\begin{definition}[Decision Function]
\label{def:decfun}
A \emph{decision function} is a mapping
$D : S \times A \to \Dset$
that assigns a disposition to each (state, action) pair.  We say $D$ is
\emph{consistent with} $\Adm$ if the following hold for all $(s,a) \in S \times A$:
\begin{enumerate}[label=(\roman*)]
  \item $D(s, a) = \Allow \Rightarrow \Adm(s, a)$;
  \item $D(s, a) = \Refuse \Rightarrow \neg\Adm(s, a)$;
  \item $\Adm(s, a) \Rightarrow D(s, a) \neq \Refuse$
        \quad (admissible states are never refused without escalation);
  \item $\neg\Adm(s, a) \Rightarrow D(s, a) \neq \Allow$
        \quad (inadmissible states are never allowed).
\end{enumerate}
Conditions (i)--(iv) together imply:
$D(s,a) = \Allow \Leftrightarrow \Adm(s,a)$ whenever $D(s,a) \neq \Escalate$.
The $\Escalate$ disposition is reserved for cases where admissibility is
indeterminate and requires external resolution.
\end{definition}

\begin{definition}[Transition Function]
\label{def:transfun}
A \emph{transition function} is a mapping
$T : S \times A \to S$
such that $T(s, a)$ gives the resulting state when action $a$ is applied in
state $s$.
\end{definition}

\begin{remark}[The role of $F$ versus composition of $D$ and $T$]
\label{rem:fvscomp}
One might observe that the type signature $F : S \times A \to \Dset \times S$
is mathematically identical to the pointwise pair $(D, T)$: given $(s, a)$,
compute $D(s,a)$ for the disposition and $T(s, a)$ for the resulting state.
Any system---including a split one---can define such a function as a
mathematical object.  The contribution of the Atomic Decision Boundary
(Definition~\ref{def:atomic}) is not in the signature but in the
\emph{indivisibility axiom}: the computation of $(d, s')$ must correspond to a
\emph{single arc} in the LTS, exposing no intermediate state to environment
actions.  In a split system, the same pair $(d, s')$ is eventually produced,
but the path from $s$ to $s'$ passes through the intermediate state $s_D$
(introduced formally in Definition~\ref{def:split} below, \S\ref{sec:split})
where environment actions can interleave.  The structural difference between
the two classes is in the topology of the LTS---specifically, whether there
exists an intermediate state---not in the domain or codomain of the function.
\end{remark}

\subsection{Atomic Decision Boundary}
\label{sec:atomic}

\begin{definition}[Atomic Decision Boundary]
\label{def:atomic}
A system satisfies an \emph{atomic decision boundary} if there exists a
function
\[
  F : S \times A \to \Dset \times S
\]
such that for every $(s, a) \in S \times A$:
\begin{align}
  F(s, a) &= (d,\, \hat{s}') \label{eq:atomic-pair}\\
  d &= D(s, a) \label{eq:atomic-decision}\\
  \hat{s}' &= (T(s, a),\, P) \quad \text{if } d = \Allow \label{eq:atomic-transition}\\
  \hat{s}' &= (s,\, P) \quad \text{if } d = \Refuse \label{eq:atomic-refuse}\\
  \hat{s}' &= (s,\, P \cup \{(s,a)\}) \quad \text{if } d = \Escalate
              \label{eq:atomic-escalate}
\end{align}
where $\hat{s}' \in \hat{S} = S \times \Req$ is the post-transition state
in the escalation-extended state space (defined formally in
Definition~\ref{def:escstate}, \S\ref{sec:escalate}).  For
$d \in \{\Allow, \Refuse\}$, the pending-request component $P$ is unchanged
by the transition; the state reduces to $S$ in those cases.
The computation of $d$ and $\hat{s}'$ corresponds to a \emph{single indivisible
transition} in the LTS: $s \xrightarrow{F(s,a)} \hat{s}'$ with no intermediate
state between the evaluation of $d$ and the determination of $\hat{s}'$.
\end{definition}

The indivisibility condition in Definition~\ref{def:atomic} is a model-level
axiom: it asserts that no environment action $e \in A_{\mathrm{env}}$ can be
interleaved between the evaluation of $D(s, a)$ and the firing of $T(s, a)$.
In terms of the LTS, the $F(s,a)$ transition is a single arc; it does not
decompose into two separate transitions with any state in between.

\begin{remark}[Implementation independence]
\label{rem:impl}
Atomicity in Definition~\ref{def:atomic} is a structural property of the
system architecture, not of any particular implementation.  Whether it is
realized via a database transaction, a monitor lock, a hardware instruction, or
a protocol contract is an implementation matter.  The definition abstracts over
all of these and asserts only that the LTS representation of the system does
not expose an intermediate state.
\end{remark}

\subsection{Split Evaluation Systems}
\label{sec:split}

\begin{definition}[Split Evaluation System]
\label{def:split}
A system is a \emph{split evaluation system} if the decision function $D$ and
the transition function $T$ are evaluated as distinct LTS transitions,
potentially separated by environment actions.  Formally, the execution of
agent action $a$ from state $s$ decomposes into:
\begin{enumerate}[label=(\arabic*)]
  \item A \emph{decision transition}: $s \xrightarrow{\mathit{dec}(a)} s_D$,
    where $s_D$ encodes the recorded disposition $d = D(s, a)$ alongside the
    current state.
  \item An optional sequence of environment transitions:
    $s_D \xrightarrow{e_1} \cdots \xrightarrow{e_k} s^*$, for $k \geq 0$.
  \item An \emph{execution transition}: $s^* \xrightarrow{\mathit{exec}(a)}
    T(s^*, a)$, contingent on $d = \Allow$.
\end{enumerate}
The states $s$ (evaluation state) and $s^*$ (execution state) may differ
whenever $k \geq 1$.
\end{definition}

The structural distinction between atomic and split systems is not a matter of
policy expressiveness, state availability, or evaluation latency.  It is a
matter of whether the LTS exposes an intermediate state between decision and
execution.  In atomic systems it does not; in split systems it does.

\subsection{Escalation Semantics}
\label{sec:escalate}

The $\Escalate$ disposition is structurally distinct from $\Refuse$: rather
than terminating the request, it suspends the transition pending external
review.  This outcome has no counterpart in classical TOCTOU analyses, which
operate in binary (permit/deny) models.  We give it a precise semantics.

\begin{definition}[Escalation State]
\label{def:escstate}
Let $\Req = \mathcal{P}_{\mathrm{fin}}(S \times A)$ denote the set of finite
sets of \emph{pending requests}.  The \emph{escalation-extended state space}
is $\hat{S} = S \times \Req$.  A state $\hat{s} = (s, P) \in \hat{S}$
encodes the current system state $s$ together with the set $P$ of pending
(state, action) pairs awaiting supervisor resolution.
\end{definition}

\begin{definition}[Escalation Transition]
\label{def:esctrans}
When $F(s, a) = (\Escalate, \hat{s}')$, the system transitions to
\[
  \hat{s}' = \bigl(s,\; P \cup \{(s, a)\}\bigr),
\]
recording the pending request $(s, a)$ with the originating state $s$, while
leaving the system state unchanged.  No transition of $T$ fires.
\end{definition}

\begin{definition}[Supervisor Resolution Function]
\label{def:resolve}
The \emph{resolution function} $\mathrm{resolve} : \hat{S} \times (S \times A)
\times \{{\Allow, \Refuse}\} \to \hat{S}$ is defined, for
$(s_{\mathrm{orig}}, a) \in P$ (precondition: the request being resolved must
be in the pending set), by:
\begin{align}
  \mathrm{resolve}\bigl((s_t, P),\; (s_{\mathrm{orig}}, a),\; \Allow\bigr)
    &= \bigl(T(s_t, a),\; P \setminus \{(s_{\mathrm{orig}}, a)\}\bigr)
    \label{eq:res-allow}\\
  \mathrm{resolve}\bigl((s_t, P),\; (s_{\mathrm{orig}}, a),\; \Refuse\bigr)
    &= \bigl(s_t,\; P \setminus \{(s_{\mathrm{orig}}, a)\}\bigr)
    \label{eq:res-refuse}
\end{align}
where $s_t$ is the \emph{current} system state at resolution time, which may
differ from $s_{\mathrm{orig}}$.  The function is undefined when
$(s_{\mathrm{orig}}, a) \notin P$; resolution of a request not in the pending
set is a protocol error.
\end{definition}

Two observations are immediate.  First, when $\Allow$ is issued at resolution,
$T$ fires in state $s_t$, not $s_{\mathrm{orig}}$.  The supervisor therefore
cannot rely on the admissibility evaluation performed in $s_{\mathrm{orig}}$;
it must re-evaluate $\Adm(s_t, a)$.  Second, the resolution function
itself is a decision point: it evaluates a disposition ($\Allow$ or $\Refuse$)
and potentially fires a transition ($T(s_t, a)$).  This is precisely the
structure of a decision boundary.

\begin{remark}[The supervisor as a second decision system]
\label{rem:supervisor}
The $\Escalate$ outcome introduces a second decision point operated by an
external supervisor, which raises a natural question: is this the same system,
or a different one?  The model is agnostic.  The supervisor may be a human
reviewer, an automated policy agent, or a recursive invocation of the same
governance pipeline.  What Corollary~\ref{cor:escalate} (to follow) establishes
is independent of this: \emph{wherever} the supervisor sits, its resolution
function $\mathrm{resolve}$ constitutes a decision boundary of its own, and
the same structural argument that applies to the primary boundary applies
there.  The escalation model does not solve the atomicity problem by deferral;
it makes explicit that the problem \emph{cannot} be avoided by deferral, only
transferred.
\end{remark}

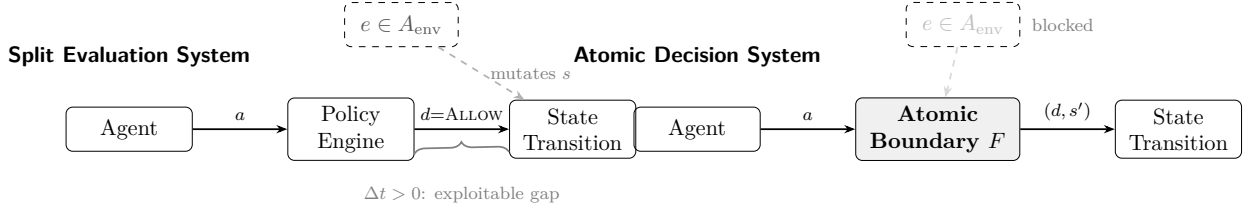
\begin{figure}[t]
\centering
\resizebox{\textwidth}{!}{%
\begin{tikzpicture}[
  node distance = 1.0cm and 1.5cm,
  box/.style    = {draw, rounded corners=3pt, minimum width=2.0cm,
                   minimum height=0.7cm, align=center, font=\small},
  envbox/.style = {draw, dashed, rounded corners=3pt, minimum width=1.8cm,
                   minimum height=0.7cm, align=center, font=\small,
                   text=gray!70!black},
  arr/.style    = {-{Stealth[length=5pt]}, thick},
  darr/.style   = {-{Stealth[length=5pt]}, thick, dashed, gray!60},
  label/.style  = {font=\small\bfseries}
]

\node[label] (splitlabel) {\textsf{Split Evaluation System}};

\node[box, below=0.5cm of splitlabel] (agent1)   {Agent};
\node[box, right=1.5cm of agent1]     (policy)   {Policy\\Engine};
\node[box, right=1.5cm of policy]     (db1)      {State\\Transition};
\node[envbox, above=0.8cm of policy,
      xshift=0.8cm]                   (env)      {$e \in A_{\mathrm{env}}$};

\draw[arr] (agent1) -- node[above,font=\scriptsize]{$a$} (policy);
\draw[arr] (policy) -- node[above,font=\scriptsize]{$d{=}\Allow$} (db1);
\draw[darr] (env)   -- node[right,font=\scriptsize,text=gray]{mutates $s$} (db1);

\draw[decorate, decoration={brace, amplitude=5pt, raise=3pt},
      gray, thick]
  (policy.south east) -- (db1.south west)
  node[midway, below=8pt, font=\scriptsize, gray]{$\Delta t > 0$: exploitable gap};

\begin{scope}[xshift=9.0cm]
\node[label] (atomlabel) {\textsf{Atomic Decision System}};

\node[box, below=0.5cm of atomlabel] (agent2) {Agent};

\node[box, right=1.5cm of agent2,
      minimum width=2.6cm,
      fill=black!6,
      label={[font=\tiny, text=black!50]south:single LTS step}]
      (kernel) {Atomic\\Boundary $F$};

\node[box, right=1.5cm of kernel] (db2) {State\\Transition};

\draw[arr] (agent2) -- node[above,font=\scriptsize]{$a$}    (kernel);
\draw[arr] (kernel) -- node[above,font=\scriptsize]{$(d,s')$} (db2);

\node[envbox, above=0.8cm of kernel,
      xshift=0.4cm, text=gray!50] (noenv) {$e \in A_{\mathrm{env}}$};
\draw[darr, gray!30] (noenv) -- (kernel);
\node[font=\scriptsize, gray, right=0.05cm of noenv] {blocked};
\end{scope}

\end{tikzpicture}}
\caption{\raggedright Architectural contrast between a split evaluation
system (left) and an atomic decision system (right).  In the split
system, environment actions $e\!\in\!A_{\mathrm{env}}$ may interleave
between the policy evaluation and the state transition, producing the
exploitable gap $\Delta t > 0$ proved in Theorem~\ref{thm:main}.
In the atomic system, $F$ determines both the decision and the resulting
state as a single indivisible LTS step; no environment action can
interleave within it.}
\label{fig:architecture}
\end{figure}

\section{The Failure of Split Systems}
\label{sec:failure}

Before stating the main theorem, we illustrate the failure mode concretely.

\begin{example}[Admissibility Violation in a Split System]
\label{ex:violation}
Return to Example~\ref{ex:running}.  The following three-step trace witnesses
the failure (steps correspond to those in the proof of Theorem~\ref{thm:main}):
\begin{enumerate}
  \item[\textbf{(1)}] The system is in state $s$ with $f \notin \mathit{locked}(s)$.
    The agent requests $\mathit{write}(f)$; the split system evaluates
    $D(s, \mathit{write}(f)) = \Allow$ and records the decision.
  \item[\textbf{(2)}] Before execution fires, another agent issues $\mathit{lock}(f)$,
    an environment action that transitions $s \to s^*$ with
    $f \in \mathit{locked}(s^*)$.
  \item[\textbf{(3)}] The original decision ($\Allow$) is already recorded.  The
    execution transition $\mathit{exec}(\mathit{write}(f))$ fires in $s^*$,
    applying $T(s^*, \mathit{write}(f))$.
\end{enumerate}
At the moment of execution, $\Adm(s^*, \mathit{write}(f)) = \mathsf{false}$:
the file is locked.  The split system has executed an inadmissible transition.

In contrast, an atomic system would evaluate $F(s^*, \mathit{write}(f))$ at
the moment of the transition and obtain $(\Refuse, s^*)$, preventing the
inadmissible execution.
\end{example}

The failure in Example~\ref{ex:violation} cannot be remedied by:
\begin{itemize}
  \item \textbf{Richer policies.}  The decision function $D$ already
    produces the correct result for every state it is given.  The problem is
    that it is given $s$, not $s^*$.

  \item \textbf{External state.}  Adding an external store (e.g., a shared
    lock registry read by $D$) reduces the window in which the gap can be
    exploited but does not reduce it to zero.  The store must be read before
    $T$ fires; between that read and the firing of $T$, the store can be
    updated.

  \item \textbf{Re-evaluation before execution.}  If the system re-evaluates
    $D$ immediately before firing $T$, the re-evaluation and the execution are
    themselves two separate transitions---which is exactly the structure of a
    split system.  Unless the re-evaluation and the execution are jointly
    atomic, the gap is merely shifted, not closed.
\end{itemize}

\section{Main Results}
\label{sec:results}

\subsection{Non-Equivalence Theorem}

We now prove the central result.  Throughout this section,
\emph{equivalence} between two systems is taken with respect to
\emph{admissibility preservation} (Definition~\ref{def:admpres}): two systems
are equivalent if and only if they agree on which execution traces preserve
admissibility under all environments satisfying Assumptions~\ref{asm:nontrivial}
and~\ref{asm:env}.  The theorem is stated for the general case; the proof is
constructive via the execution trace in Example~\ref{ex:violation}.

\begin{theorem}[Non-Equivalence of Split and Atomic Systems]
\label{thm:main}
Let $M_{\mathrm{split}}$ be any split evaluation system and let
$M_{\mathrm{atom}}$ be any atomic decision boundary system, both defined
over the same state space $S$, action set $A$, environment actions
$A_{\mathrm{env}}$, admissibility predicate $\Adm$, and transition function
$T$.  Under Assumptions~\ref{asm:nontrivial} and~\ref{asm:env}, there exists an execution trace
$\sigma^* \in \Traces(M_{\mathrm{split}})$ such that:
\begin{enumerate}[label=(\roman*)]
  \item $\sigma^*$ contains a transition $s^* \xrightarrow{\mathit{exec}(a)}
    T(s^*, a)$ with $\Adm(s^*, a) = \mathsf{false}$.
  \item No corresponding trace exists in $\Traces(M_{\mathrm{atom}})$ that
    commits $T(\cdot, a)$ in a state where $\Adm(\cdot, a) = \mathsf{false}$.
\end{enumerate}
Consequently, $M_{\mathrm{split}}$ does not preserve admissibility, while
$M_{\mathrm{atom}}$ does.
\end{theorem}

\begin{proof}
By Assumption~\ref{asm:nontrivial}(i), there exist $s \in S$ and $a \in A$
with $\Adm(s, a) = \mathsf{true}$.  By Assumption~\ref{asm:nontrivial}(ii),
there exists $e \in A_{\mathrm{env}}$ and $s^* \in S$ with $s \xrightarrow{e}
s^*$ and $\Adm(s^*, a) = \mathsf{false}$.

\paragraph{Constructing $\sigma^*$.}
Consider the following trace in $M_{\mathrm{split}}$:
\[
  \sigma^* :\quad
  s \;\xrightarrow{\mathit{dec}(a)}\; s_D
    \;\xrightarrow{e}\; s^*
    \;\xrightarrow{\mathit{exec}(a)}\; T(s^*, a).
\]
\begin{enumerate}
  \item The decision transition $\mathit{dec}(a)$ fires in state $s$.  By
    Assumption~\ref{asm:nontrivial}(i), $D(s, a) = \Allow$.  The system
    records $\Allow$ and transitions to $s_D$ (encoding the recorded decision
    alongside the current state; the underlying system state is still $s$).

  \item The environment transition $e$ fires, producing $s^*$ with
    $\Adm(s^*, a) = \mathsf{false}$.

  \item The execution transition $\mathit{exec}(a)$ fires in $s^*$.  In
    $M_{\mathrm{split}}$, execution is triggered by the recorded disposition
    $\Allow$, which was evaluated in $s$.  The split system applies
    $T(s^*, a)$, yielding the inadmissible transition.
\end{enumerate}

By construction, $\sigma^* \in \Traces(M_{\mathrm{split}})$ (all three
transitions are valid in a split system), and the transition at step~(3) is
inadmissible in $s^*$.  This establishes~(i).

\paragraph{Atomic system behavior.}
In $M_{\mathrm{atom}}$, the agent action $a$ corresponds to a single
indivisible LTS transition $F(s', a)$ from whatever state $s'$ holds at the
moment the agent acts.

If the agent acts \emph{before} the environment transition $e$ fires (system
in state $s$): $F(s, a) = (\Allow, T(s, a))$.  The transition completes
atomically; $e$ cannot interleave \emph{within} the single LTS step.  The
resulting state is $T(s, a)$, produced from $s$ where $\Adm(s, a)$ holds.

If the agent acts \emph{after} $e$ fires (system in state $s^*$):
$D(s^*, a) \in \{\Refuse, \Escalate\}$, since $\Adm(s^*, a) = \mathsf{false}$
and $D$ is consistent with $\Adm$ (condition~(iv): $\neg\Adm \Rightarrow D \neq \Allow$).
In either case, no transition of $T$ fires: $\Refuse$ keeps the state unchanged,
and $\Escalate$ records the pending request without firing $T$
(Definition~\ref{def:esctrans}).

In neither case does $M_{\mathrm{atom}}$ commit $T(\cdot, a)$ in a state
where admissibility fails.  This establishes~(ii).

\paragraph{Conclusion.}
$M_{\mathrm{split}}$ produces an inadmissible transition under $\sigma^*$;
$M_{\mathrm{atom}}$ does not.  Therefore $M_{\mathrm{split}}$ does not satisfy
Definition~\ref{def:admpres} (admissibility preservation) under all execution
traces.  No construction applied to $M_{\mathrm{split}}$ can restore
admissibility preservation without introducing an atomic boundary: any
mechanism that prevents the violation in $\sigma^*$ must either
(a) block environment actions from interleaving between $\mathit{dec}(a)$ and
$\mathit{exec}(a)$---which is precisely Definition~\ref{def:atomic}---or
(b) re-evaluate $D$ at execution time jointly with $T$ as a single
indivisible step---which is again Definition~\ref{def:atomic}.  In both cases,
the fix is the atomic boundary condition itself.
\end{proof}

\begin{remark}[Scope and boundary conditions of the impossibility]
\label{rem:scope}
Theorem~\ref{thm:main} is a conditional result: it holds under
Assumptions~\ref{asm:nontrivial} and~\ref{asm:env}.  In a \emph{sequential}
system where no concurrent environment actions occur ($k = 0$ always in
Definition~\ref{def:split}), the trace $\sigma^*$ is unreachable and the
gap cannot be exploited.  Similarly, if mutual exclusion is enforced at the
operating-system or runtime level over the \emph{entire} interval between
$\mathit{dec}(a)$ and $\mathit{exec}(a)$, the split system becomes effectively
atomic within that exclusion scope---but this is precisely an implementation
of the atomic boundary (see Remark~\ref{rem:impl}).

Two common objections deserve direct answers.  \emph{First}: ``What about
locks, transactions, CAS, or linearizable storage?''  These are valid
implementation mechanisms for the atomic boundary; they are not alternatives
to it.  A system that wraps both the decision call and the execution call inside
a single database transaction, or acquires a lock before $\mathit{dec}(a)$ and
releases it only after $\mathit{exec}(a)$, satisfies Definition~\ref{def:atomic}
by construction---because no environment action can fire between the two.
\emph{Second}: ``Isn't this just data atomicity?''  No.  Data atomicity (2PC,
MVCC) ensures that a data record is fully written or not written at all.
Admissibility atomicity ensures that the \emph{policy check} and the
\emph{data write} are evaluated against the same system-state snapshot.
A 2PC-coordinated write preceded by an OPA policy call is still a split system
in our model (\S\ref{sec:dist}).
\end{remark}

\subsection{Corollaries}

\begin{corollary}[Impossibility of Execution-Time Guarantees in Split Systems]
\label{cor:impossibility}
A split evaluation system cannot guarantee admissibility at execution time
under all execution traces.
\end{corollary}

\begin{proof}
Immediate from Theorem~\ref{thm:main}: the trace $\sigma^*$ is a
constructive witness to a violation.  Since $\sigma^*$ is a valid trace of
any split system satisfying Assumption~\ref{asm:nontrivial}, no such system
can claim universal admissibility preservation.
\end{proof}

\begin{corollary}[External State Does Not Restore Atomicity]
\label{cor:extstate}
Augmenting a split system with access to external state does not restore
equivalence with an atomic system unless evaluation and the resulting
state transition are jointly enforced as a single indivisible step.
\end{corollary}

\begin{proof}
External state enriches the domain of the decision function $D$: let
$D' : S \times A \times \mathcal{E} \to \Dset$ where $\mathcal{E}$ is the
external state.  The augmented system still decomposes into a decision
transition (invoking $D'$ with the external state read at that moment) and a
subsequent execution transition.  The environment action $e$ in $\sigma^*$ may
act on the external state or the system state between these transitions; the
structural gap is unchanged.  Formally, $D'(s, a, \epsilon)$ with
$\epsilon = \epsilon_s$ (the external state corresponding to $s$) may return
$\Allow$, while by the time $T$ fires the external state has evolved to
$\epsilon_{s^*}$ reflecting the lock acquired by $e$.  The gap in
$\sigma^*$ is not closed.  Atomic coupling of $D'$ and $T$ is required.
\end{proof}

\begin{corollary}[Admissibility is a Property of Execution]
\label{cor:execprop}
Admissibility cannot be reduced to a property of evaluation alone; it is a
property of the state at the moment the transition is committed.
\end{corollary}

\begin{proof}
In the trace $\sigma^*$, the same action $a$ is admissible in state $s$
(evaluation state) and inadmissible in state $s^*$ (execution state).  Since
execution commits the transition in $s^*$, admissibility must be determined
with respect to $s^*$.  Any governance guarantee that is computed only at
$s$ is therefore insufficient.
\end{proof}

\subsection{The Escalation Closure Requirement}

\begin{corollary}[Escalation Closure Requirement]
\label{cor:escalate}
An $\Escalate$ outcome preserves the atomic decision boundary guarantee if
and only if the supervisor resolution function $\mathrm{resolve}$ is itself
atomic: the evaluation of $\Adm(s_t, a)$ at resolution time and the
application of $T(s_t, a)$ must be jointly determined as a single indivisible
step.
\end{corollary}

\begin{proof}
Let the system issue $\Escalate$ for action $a$ in state $s$, recording the
pending request $(s, a)$ per Definition~\ref{def:esctrans}.

\emph{(Necessity.)}  Suppose $\mathrm{resolve}$ is split: the supervisor
evaluates $D(s_t, a)$ in state $s_t$, and then fires $T$ separately.

We verify that Assumptions~\ref{asm:nontrivial} and~\ref{asm:env} hold for
the resolution sub-process operating on $\hat{S} = S \times \Req$.
Assumption~\ref{asm:nontrivial}(i) holds because $(s_t, a)$ has
$\Adm(s_t, a) = \mathsf{true}$ (otherwise the supervisor would not issue
$\Allow$) and $D(s_t, a) = \Allow$ (this is the split supervisor's decision).
Assumption~\ref{asm:nontrivial}(ii) holds because the system state $s_t$
(the first component of $\hat{S}$) can be modified by environment actions
between the supervisor's decision transition and the execution of $T$;
the same Assumption~\ref{asm:env} that applies to the primary system applies
to the projected $S$-component of $\hat{S}$.
Assumption~\ref{asm:env}(i)--(iii) carry over verbatim to the projected
$S$-component: environment actions are uncontrollable, may interleave
between the supervisor's evaluation and the firing of $T$, and the
supervisor observes only the current state $s_t$ at the time of evaluation.

By Theorem~\ref{thm:main} applied to this resolution sub-process, there
exists a trace in which the supervisor's $\Allow$ is computed in $s_t$ but
$T$ fires in $s^{**}$ with $\Adm(s^{**}, a) = \mathsf{false}$.
Thus, $\mathrm{resolve}$ admits an inadmissible transition.

\emph{(Sufficiency.)}  Suppose $\mathrm{resolve}$ is atomic: the evaluation
$D(s_t, a)$ and the application $T(s_t, a)$ are a single indivisible step in
the LTS.  Then by Definition~\ref{def:atomic} applied to $\mathrm{resolve}$,
the resolution satisfies the atomic decision boundary condition, and
Theorem~\ref{thm:main} guarantees no inadmissible transition at resolution.
\end{proof}

Corollary~\ref{cor:escalate} has an important practical consequence: the
$\Escalate$ outcome does not relax the atomicity obligation; it \emph{transfers}
it from the primary enforcement point to the supervisor.  A chain of escalations
each requiring their own atomic resolution corresponds precisely to the notion
of a supervisory hierarchy in discrete event systems~\cite{ramadge87}.

\section{Mapping to Existing Systems}
\label{sec:mapping}

Table~\ref{tab:mapping} classifies common governance mechanisms according to
the distinction established in \S\ref{sec:dbm}.  The classification is
structural and does not reflect comparisons of expressiveness or performance.

\begin{table}[h]
\centering
\caption{Structural classification of existing governance mechanisms.}
\label{tab:mapping}
\begin{tabular}{llp{6.8cm}}
\toprule
\textbf{System} & \textbf{Class} & \textbf{Basis for classification} \\
\midrule
RBAC~\cite{sandhu96}
  & Split
  & Permission check and execution are distinct operations \\[2pt]
ABAC~\cite{hru76}
  & Split
  & Attribute evaluation external to transition; attributes may change \\[2pt]
OPA~\cite{opa}
  & Split
  & Policy evaluation external to state transition \\[2pt]
Cedar~\cite{cedar}
  & Split
  & Authorization decision external to the enforcing runtime \\[2pt]
OPA + Redis
  & Split
  & External state enriches $D$; structural gap remains
    (Cor.~\ref{cor:extstate}) \\[2pt]
AWS IAM
  & Split
  & IAM evaluation precedes API call execution; no shared-snapshot
    guarantee \\[2pt]
Kubernetes Admission
  & Partially atomic
  & Webhook admits before object creation, but cross-resource state
    (e.g., a \texttt{ResourceQuota} modified concurrently) is not
    jointly atomic with the admission decision
    (Remark~\ref{rem:partial-atomic}) \\[2pt]
ACP~\cite{acp26}
  & Atomic
  & Decision and state mutation jointly enforced via single-use
    Execution Token \\
\bottomrule
\end{tabular}
\end{table}

\paragraph{Role-Based Access Control (RBAC).}
RBAC assigns permissions to roles and roles to principals.  At access time,
the system evaluates whether the requesting principal holds a role that
permits the requested operation.  The evaluation produces a disposition;
the operation is then executed by the application.  The RBAC system has no
mechanism to ensure that the role assignment---or the system state on which
the permission check depends---remains unchanged between evaluation and
execution.  RBAC implements:
\[
  d = D_{\mathrm{RBAC}}(s, a), \quad s' = T(s, a),
\]
as separate operations. RBAC is a split evaluation system.

\emph{Concrete trace.}  In state $s$, an agent holds the \textsc{Editor} role
and requests $\mathit{write}(r)$.  The RBAC check returns $d = \Allow$.
Before the write executes, an administrator revokes the role via environment
action $e$, transitioning to $s^*$ where the agent is no longer an
\textsc{Editor} and $\Adm(s^*, \mathit{write}(r)) = \mathsf{false}$.  The
recorded $\Allow$ nonetheless triggers execution of $T(s^*, \mathit{write}(r))$.
This is precisely the trace $\sigma^*$ of Theorem~\ref{thm:main}.

\paragraph{Open Policy Agent (OPA).}
OPA evaluates policies expressed in Rego against an input document.  The
evaluation produces a decision; execution is performed externally by the
calling system.  In its standard deployment model:
\[
  d = D_{\mathrm{OPA}}(s, a),
\]
with $T(s, a)$ executed outside the policy engine.  Even when OPA is
augmented with external data injected at query time (e.g., via the bundle API
or a data source), the execution transition remains external.  By
Corollary~\ref{cor:extstate}, this augmentation does not restore atomicity.
OPA is a split evaluation system.

\emph{Concrete trace.}  An OPA bundle is evaluated against input document
encoding state $s$ (e.g., current API quota usage $= q < q_{\max}$) and
returns $d = \Allow$.  Between the policy call and the application's write
execution, a concurrent request exhausts the remaining quota, producing $s^*$
with $\Adm(s^*, \mathit{write}) = \mathsf{false}$.  The application, holding
the recorded $\Allow$, proceeds.  Injecting a real-time Redis data source into
OPA narrows the window but does not close it: the Redis read is itself a
distinct LTS transition exposed to interleaving between $\mathit{dec}(a)$
and $\mathit{exec}(a)$ (Corollary~\ref{cor:extstate}).

\paragraph{ACP (Agent Control Protocol).}
ACP implements admission control as a stateful runtime mechanism in which
the decision and state mutation are jointly enforced at the execution
boundary~\cite{acp26}.  Each agent action must obtain an Execution Token---a
single-use cryptographic proof of admission---before the action is permitted
to proceed.  The token is issued and consumed atomically relative to the
ledger state, ensuring that the system state against which admissibility is
evaluated is the same state against which the transition commits.

ACP further implements the $\Escalate$ outcome as a pending-review state
tracked in the Audit Ledger (ACP-LEDGER-1.3), where subsequent resolution
passes through the same admission pipeline.  By Corollary~\ref{cor:escalate},
this satisfies the escalation closure requirement.

Formally:
\[
  F_{\mathrm{ACP}}(s, a) = (d, s'),
\]
where $d$ and $s'$ are determined by a single evaluate-then-mutate execution
contract~\cite{acp26}.  ACP is an atomic decision boundary system.

\paragraph{AWS IAM.}
AWS IAM evaluates authorization requests against a set of policies attached
to principals, resources, and conditions.  An IAM authorization call
(e.g., to the IAM policy simulator or to the underlying STS service)
produces a decision of $\Allow$ or $\Refuse$.  The resulting API call is
executed by the service (e.g., S3, EC2) independently of the IAM evaluation
context.  No shared-snapshot guarantee exists between the IAM policy
evaluation context and the service's execution context: condition keys
(e.g., \texttt{aws:CurrentTime}, \texttt{ec2:ResourceTag}) are evaluated
at IAM call time, while the resource state at execution time may differ.

\emph{Concrete trace.}  An agent requests an S3 write.  IAM evaluates the
policy in state $s$ where the bucket's ACL permits writes from the agent's
role, returning $d = \Allow$.  Before the write reaches S3, a concurrent
process changes the bucket policy (environment action $e$), producing $s^*$
where the write is inadmissible.  The S3 service executes $T(s^*, \mathit{write})$
regardless, since it received the IAM authorization token produced in $s$.
This is the trace $\sigma^*$ of Theorem~\ref{thm:main}.  AWS IAM is a split
evaluation system.

\paragraph{Cedar.}
Cedar is a policy language and evaluation engine designed for authorization
decisions in service-oriented architectures~\cite{cedar}.  Like OPA, Cedar
evaluates policies against entities (principal, action, resource, context)
and returns a decision.  The enforcement of that decision---the actual
execution of the requested operation---occurs outside Cedar, in the calling
application or service.  Cedar provides no mechanism to jointly enforce
its authorization decision with the subsequent state transition; the two
remain separate LTS transitions.  By Corollary~\ref{cor:extstate},
augmenting Cedar with an entity store does not restore atomicity: the
entity store is read at evaluation time, and the application may mutate
relevant entity attributes between Cedar's decision and the action's execution.
Cedar is a split evaluation system.

\begin{remark}[Partially atomic systems]
\label{rem:partial-atomic}
Table~\ref{tab:mapping} uses the term \emph{partially atomic} for
Kubernetes admission webhooks.  This denotes a system that satisfies
the atomic decision boundary condition (Definition~\ref{def:atomic}) for
a restricted subset of state---specifically, single-object lifecycle
operations---but fails atomicity when the relevant admissibility predicate
$\Adm$ depends on shared or cross-resource state that is not included in
the joint evaluation-and-execution snapshot.  Formally: let $S$ decompose
as $S = S_{\mathrm{local}} \times S_{\mathrm{global}}$.  If the system is
atomic with respect to $S_{\mathrm{local}}$ but split with respect to
$S_{\mathrm{global}}$, it is partially atomic.  When $\Adm$ depends only
on $S_{\mathrm{local}}$, the system provides full guarantees; when $\Adm$
depends on $S_{\mathrm{global}}$, Theorem~\ref{thm:main} applies to the
$S_{\mathrm{global}}$ component.
\end{remark}

\paragraph{Kubernetes Admission (Partially Atomic).}
Kubernetes admission webhooks intercept API server requests before the object
is written to etcd, providing a form of pre-execution evaluation.  Within a
single object's lifecycle, the webhook decision and the subsequent write to
etcd are reasonably coupled.  However, Kubernetes does not enforce atomicity
across resources: consider an agent Pod requesting a batch job that is
admissible given the current \texttt{ResourceQuota} state.  Between the
webhook's admission decision and the Pod's binding to a node, a concurrent
namespace operation updates the \texttt{ResourceQuota} to a value that would
have caused refusal.  The transition fires in an inadmissible cross-resource
state.  In terms of our model, the state $s$ relevant to $\Adm$ encodes
cluster-wide resource state; the admission decision is evaluated in $s$ but
the execution commits in a state $s^*$ that the webhook did not observe.  The
system is partially atomic (per-object) but split with respect to the full
admissibility predicate over shared cluster state.

\section{Discussion}
\label{sec:discussion}

\subsection{Relationship to Distributed Transactions}
\label{sec:dist}

Two-phase commit (2PC) and multiversion concurrency control (MVCC) achieve
atomicity for \emph{data transitions}: they ensure that a database record is
either fully updated or not updated at all.  This is a different plane from
the admissibility atomicity we define.

Admissibility is a predicate over the system state at transition time; its
enforcement requires that the policy check and the transition commit occur in
the same snapshot of that state.  2PC coordinates multiple data stores to
commit together; it does not ensure that the policy check preceding the
commit was evaluated in the same snapshot.  A system can run an OPA policy
check followed by a 2PC-coordinated database update and still be a split
system in our model: the OPA call and the 2PC commit are separate operations
with an exploitable window between them.

An atomic decision boundary requires that the policy evaluation and the state
commit are \emph{the same operation}---not merely that they are coordinated.

\paragraph{Relationship to linearizability.}
Linearizability~\cite{herlihy90} is the standard correctness criterion for
concurrent objects: an operation appears to take effect instantaneously at
some point between its invocation and its response.  This is a property of
\emph{data operations} on shared objects.  The atomic decision boundary
operates at a distinct level of abstraction: it is a property of
\emph{governance decisions} over state transitions.

A linearizable data store guarantees that each read or write appears atomic
in isolation; it does not guarantee that a policy check and the write it
governs were evaluated against the same state snapshot.  A system that runs a
policy evaluation followed by a linearizable write is still a split system in
our model: the policy call and the write are separate operations with an
exploitable window between them.

These are complementary guarantees, not competing ones.  A system achieves
the atomic decision boundary by implementing the policy check and the
resulting state transition as a \emph{single} linearizable operation---
precisely collapsing the gap.  Linearizability specifies what it means for a
single operation to be atomic; the atomic decision boundary specifies
\emph{which operations must be made atomic together}.

\subsection{What Atomicity Does Not Prescribe}

Our results characterize a \emph{necessary} structural condition for
admissibility guarantees.  They do not:
\begin{itemize}
  \item prescribe a specific implementation of the atomic boundary (hardware
    CAS, database transaction, protocol contract, or otherwise);
  \item claim that atomic systems are immune to all governance failures---
    behavioral drift operating below the enforcement threshold is addressed
    separately~\cite{iml26};
  \item argue that split systems have no valid use.  Split systems can
    provide probabilistic assurance, reduce latency, or serve in contexts
    where admissibility is defined conservatively.  The theorem establishes
    only that they cannot provide the universal guarantee.
\end{itemize}

\subsection{Concurrency Models}

The trace $\sigma^*$ in Theorem~\ref{thm:main} assumes an adversarial
environment: the environment action $e$ is assumed to be able to fire between
$\mathit{dec}(a)$ and $\mathit{exec}(a)$.  In a stochastic model, $e$ fires
with some probability $p > 0$; the guarantee degrades from universal to
probabilistic.  In a sequential model where no concurrency is possible, the
gap does not arise.  The relevance of Theorem~\ref{thm:main} is therefore
proportional to the degree of concurrency in the deployment environment.  For
modern multi-agent systems operating over shared state, $p > 0$ is the
realistic assumption.

\subsection{Relationship to Formal Verification}

TLA+ and other model-checking frameworks verify safety and liveness properties
of system models.  In a TLA+ model of a split system, the gap in $\sigma^*$
can be made explicit: if environment actions are part of the model, the
checker will find traces that violate admissibility invariants.  This paper
provides the structural characterization that motivates where to look and why
no refinement of the split architecture can eliminate the counterexample trace.
The ACP specification~\cite{acp26} is model-checked with TLC and verifies
safety invariants that hold precisely because ACP enforces an atomic boundary.

\section{Related Work}
\label{sec:related}

\paragraph{Time-of-Check/Time-of-Use (TOCTOU).}
The TOCTOU pattern---in which a security check and the operation it governs
are separated in time---has been studied extensively in operating systems
security~\cite{bishop96, tsafrir08}.  The classic form involves a file system
race: a process checks file permissions and then opens the file; an attacker
replaces the file between check and open.  Our model generalizes TOCTOU from
OS primitives to governance systems, with three key extensions.

First, we operate at the level of \emph{state transitions} rather than file
handles: the relevant state is the full system state encoding all properties
on which admissibility depends, not merely a permission bit.  Second, we add
the $\Escalate$ outcome, which TOCTOU analyses do not consider.  Classical
TOCTOU mitigation strategies (e.g., using \texttt{openat} or atomic
rename) work by removing the gap; our framework shows that the same principle
applies to governance systems and that the gap cannot be closed by
re-evaluation alone.  Third, we provide a positive characterization of systems
that are immune to the gap (atomic systems), not only a characterization of
vulnerable systems.

\paragraph{Supervisory Control of Discrete Event Systems.}
Ramadge and Wonham~\cite{ramadge87} introduced the formal theory of
supervisory control, in which a supervisor enables or disables events in a
plant modeled as a DES.  Their framework distinguishes controllable events
(which the supervisor can prevent) from uncontrollable events (which it
cannot).  Our framework is related: agent actions correspond to controllable
events; environment actions correspond to uncontrollable ones.  The
$\Escalate$ outcome corresponds to the supervisor suspending a controllable
event pending further observation---a concept formalized in the DES literature
as \emph{conditional enabling}~\cite{ramadge87}.  Corollary~\ref{cor:escalate}
is the governance-system analogue of the requirement that supervisory
decisions be enforced at the moment of event firing.

\paragraph{Access Control.}
The access matrix model of Harrison, Ruzzo, and Ullman~\cite{hru76} is the
classical foundation of access control theory.  RBAC~\cite{sandhu96} and
attribute-based access control (ABAC) elaborate on it.  These models define
what operations principals are permitted to perform; they do not model the
temporal relationship between the permission check and the operation
execution.  Our work complements access control theory by characterizing when
that temporal relationship matters for correctness.

\paragraph{Reference Monitors.}
Anderson~\cite{anderson72} defined the reference monitor as a component that
must be invoked on every access, must be tamper-proof, and must be small
enough to be verified.  Our atomic decision boundary captures the first and
third properties at a formal level: the boundary is invoked at the moment of
transition (complete mediation) and its definition is minimal (a single
function $F$).  The ACP protocol~\cite{acp26} is explicitly designed as a
decentralized reference monitor; the present paper provides the formal basis
for why complete mediation requires atomicity.

\paragraph{Runtime Enforcement and Security Automata.}
Schneider~\cite{schneider00} characterized which security properties are
enforceable by inline reference monitors modeled as security automata: a
property is enforceable if and only if its violations are detectable in finite
traces (i.e., it is a safety property).  The admissibility preservation
requirement of Definition~\ref{def:admpres} is a safety property in precisely
this sense---each violation is witnessed by the finite trace $\sigma^*$.
Schneider's framework assumes the monitor can observe and halt executions; our
work identifies the structural condition under which that halt can be guaranteed
\emph{atomically with respect to the state transition being governed}.  A
monitor that halts asynchronously with the transition is still a split system
in our model.

Ligatti, Bauer, and Walker~\cite{ligatti05} extended security automata to
\emph{edit automata}, which can suppress or insert events to enforce a broader
class of policies.  The $\Escalate$ outcome in our model corresponds
structurally to suppression: the transition is prevented pending external
resolution.  The Escalation Closure Requirement (Corollary~\ref{cor:escalate})
is the governance-layer analogue of the requirement that edit operations
themselves be applied consistently with the current state of the running
system---a condition edit automata assume implicitly but do not formalize at
the level of decision boundaries.

\paragraph{Policy Engines.}
OPA~\cite{opa} and Cedar~\cite{cedar} provide expressive policy languages for
evaluating authorization decisions.  These systems are powerful within their
class (split evaluation) and our taxonomy does not imply that they are
deficient---only that they cannot provide execution-time admissibility
guarantees by construction.  Hybrid architectures in which a policy engine
provides the decision function $D$ inside an atomic enforcement layer are
consistent with our model; the atomicity requirement falls on the enclosing
layer, not on the policy engine itself.

\paragraph{Temporal Logic and Model Checking.}
Safety properties in temporal logic (e.g., $\square \neg \mathrm{bad}$) can
express admissibility requirements.  Model checkers verify whether a system
satisfies such properties.  Our contribution is complementary: we identify
which \emph{system architectures} can satisfy admissibility properties
universally (atomic) and which cannot (split), independent of any particular
property specification.

\section{Conclusion}
\label{sec:conclusion}

We have introduced the \emph{atomic decision boundary} as a structural
requirement for governance systems that must enforce admissibility at execution
time.  Using a formal LTS model, we proved that split evaluation systems
cannot provide this guarantee under all execution traces, and showed that
no enrichment of the split architecture---whether through additional policies,
external state, or re-evaluation---can close the structural gap without
introducing atomicity.

The $\Escalate$ outcome extends classical binary decision models and captures
supervisor-mediated governance.  We proved that escalation transfers, rather
than removes, the atomicity obligation: any system that resolves an
$\Escalate$ must itself be atomic at the point of resolution.

The core result is direct: \emph{admissibility is not a property of
evaluation---it is a property of execution.}  Under
Assumptions~\ref{asm:nontrivial} and~\ref{asm:env}, no system that separates
evaluation from execution can guarantee admissibility at execution time; only
a system that couples them as a single indivisible step at the decision
boundary can provide this guarantee.

This characterization serves as the formal foundation for the subsequent papers
in this series of six.  The Agent Control Protocol~(ACP)~\cite{acp26}
(Paper~1) instantiates $F$ as a concrete atomic boundary: each agent action
acquires a single-use Execution Token that is issued and consumed in the same
ledger operation, ensuring that the state against which admissibility is
evaluated is precisely the state against which the transition commits.  The
Invariant Measurement Layer~\cite{iml26} (Paper~2) operates \emph{above} the
atomic enforcement layer---taking atomicity as given---and detects behavioral
drift that accumulates across sequences of individually admissible actions,
remaining invisible to per-action enforcement signals.  Paper~3~\cite{fairgov26} extends the framework to the fair allocation of
governance decisions across competing agents, proving that correct enforcement
does not imply fair allocation and characterising the allocation layer
explicitly.  Paper~4~\cite{fernandez2026comp} proves that the four layers are
\emph{irreducible}: no strict subset can replicate the governance guarantees of
the full stack, and the composition itself produces emergent properties absent
from any single layer.  Paper~5~\cite{fernandez2026ram} closes the series by
addressing the runtime question left open by Papers~0--4: given that
observability is incomplete (Paper~2) and the architecture is irreducible
(Paper~4), the Reconstructive Authority Model (RAM) provides the operational
mechanism for determining whether execution authority can be constructed from
the current real state---and halts or narrows privileges when it cannot.

\medskip
\noindent\textit{Note on terminology.}\enspace
Throughout this paper the non-admissible outcome is written $\Refuse$.
Subsequent papers in the series (Papers~3--5) adopt the synonym
$\texttt{Deny}$, following ACP's \texttt{DENIED} protocol response
(Paper~1).  The decision domain
$\mathcal{D} = \{\Allow, \Escalate, \Refuse\}$ is the same object across the
series.

All six papers rest on the structural guarantee established here: that the
boundary they operate relative to is, in fact, atomic.

\bibliographystyle{plainnat}
\bibliography{references}

@misc{acp26,
  author       = {Fernandez, Marcelo},
  title        = {Agent Control Protocol: Admission Control for Agent Actions},
  year         = {2026},
  doi          = {10.5281/zenodo.19672575},
  howpublished = {\url{https://arxiv.org/abs/2603.18829}},
  note         = {arXiv:2603.18829 [cs.CR], DOI:10.5281/zenodo.19672575}
}

@misc{iml26,
  author       = {Fernandez, Marcelo},
  title        = {From Admission to Invariants: Measuring Deviation in
                 Delegated Agent Systems},
  year         = {2026},
  doi          = {10.5281/zenodo.19672589},
  howpublished = {\url{https://doi.org/10.5281/zenodo.19672589}},
  note         = {Zenodo. DOI: 10.5281/zenodo.19672589. arXiv:2604.17517}
}

@misc{fairgov26,
  author       = {Fernandez, Marcelo},
  title        = {Fair Atomic Governance: Allocating Decision Boundaries under
                 Shared Resource Constraints in Multi-Agent Systems},
  year         = {2026},
  doi          = {10.5281/zenodo.19672597},
  howpublished = {\url{https://doi.org/10.5281/zenodo.19672597}},
  note         = {Zenodo. DOI: 10.5281/zenodo.19672597}
}

@misc{fernandez2026comp,
  author       = {Fernandez, Marcelo},
  title        = {Irreducible Multi-Scale Governance: Composition and Limits of
                 Atomic Admission Systems},
  year         = {2026},
  doi          = {10.5281/zenodo.19672608},
  howpublished = {\url{https://doi.org/10.5281/zenodo.19672608}},
  note         = {Zenodo. DOI: 10.5281/zenodo.19672608}
}

@misc{fernandez2026ram,
  author       = {Fernandez, Marcelo},
  title        = {Reconstructive Authority Model: Runtime Execution Validity
                 Under Partial Observability},
  year         = {2026},
  doi          = {10.5281/zenodo.19669430},
  howpublished = {\url{https://doi.org/10.5281/zenodo.19669430}},
  note         = {Agent Governance Series, Paper~5. Zenodo. DOI: 10.5281/zenodo.19669430}
}

@article{herlihy90,
  author  = {Herlihy, Maurice P. and Wing, Jeannette M.},
  title   = {Linearizability: A Correctness Condition for Concurrent Objects},
  journal = {ACM Transactions on Programming Languages and Systems},
  volume  = {12},
  number  = {3},
  pages   = {463--492},
  year    = {1990}
}

@article{schneider00,
  author  = {Schneider, Fred B.},
  title   = {Enforceable Security Policies},
  journal = {ACM Transactions on Information and System Security},
  volume  = {3},
  number  = {1},
  pages   = {30--50},
  year    = {2000}
}

@article{ligatti05,
  author  = {Ligatti, Jay and Bauer, Lujo and Walker, David},
  title   = {Edit Automata: Enforcement Mechanisms for Run-time Security
             Policies},
  journal = {International Journal of Information Security},
  volume  = {4},
  number  = {1--2},
  pages   = {2--16},
  year    = {2005}
}

@techreport{anderson72,
  author      = {Anderson, James P.},
  title       = {Computer Security Technology Planning Study},
  institution = {USAF Electronic Systems Division},
  number      = {ESD-TR-73-51},
  year        = {1972}
}

@techreport{bishop96,
  author      = {Bishop, Matt and Bailey, David},
  title       = {A Critical Analysis of Vulnerability Taxonomies},
  institution = {University of California at Davis},
  number      = {CSE-96-11},
  year        = {1996}
}

@inproceedings{tsafrir08,
  author    = {Tsafrir, Dan and Hertz, Tomer and Wagner, David and Da Silva, Dilma},
  title     = {Portably Solving File {TOCTTOU} Races with Hardness Amplification},
  booktitle = {USENIX Conference on File and Storage Technologies (FAST)},
  year      = {2008}
}

@article{hru76,
  author  = {Harrison, Michael A. and Ruzzo, Walter L. and Ullman, Jeffrey D.},
  title   = {Protection in Operating Systems},
  journal = {Communications of the ACM},
  volume  = {19},
  number  = {8},
  pages   = {461--471},
  year    = {1976}
}

@article{sandhu96,
  author  = {Sandhu, Ravi S. and Coyne, Edward J. and Feinstein, Hal L.
             and Youman, Charles E.},
  title   = {Role-Based Access Control Models},
  journal = {IEEE Computer},
  volume  = {29},
  number  = {2},
  pages   = {38--47},
  year    = {1996}
}

@article{ramadge87,
  author  = {Ramadge, Peter J. and Wonham, W. Murray},
  title   = {Supervisory Control of a Class of Discrete Event Processes},
  journal = {SIAM Journal on Control and Optimization},
  volume  = {25},
  number  = {1},
  pages   = {206--230},
  year    = {1987}
}

@inproceedings{cedar,
  author    = {Cutler, John and others},
  title     = {{Cedar}: A New Policy Language},
  booktitle = {USENIX Security Symposium},
  year      = {2023}
}

@misc{opa,
  key          = {OPA},
  author       = {{Open Policy Agent Contributors}},
  title        = {Open Policy Agent},
  year         = {2026},
  howpublished = {\url{https://www.openpolicyagent.org}},
  note         = {Accessed April 2026}
}

\end{document}